\documentclass[conference]{IEEEtran}
\IEEEoverridecommandlockouts
\usepackage{cite}
\usepackage{amsmath,amssymb,amsfonts}
\usepackage{algorithm}
\usepackage[noend]{algorithmic}
\usepackage{graphicx}
\usepackage{textcomp}
\usepackage{xcolor}

\usepackage[utf8]{inputenc} 
\usepackage[T1]{fontenc}    
\usepackage{url}            
\usepackage{booktabs}       
\usepackage{amsfonts}       
\usepackage{nicefrac}       
\usepackage{microtype}      
\usepackage{xcolor}         

\usepackage{verbatim}
\usepackage{graphicx}
\usepackage{textcomp}
\usepackage{cases}
\usepackage{amsmath,amsfonts}
\usepackage{fancyhdr,graphicx}
\usepackage{subfigure}
\usepackage[utf8]{inputenc}
\usepackage{soul}
\usepackage{multirow}
\usepackage{fontenc}
\usepackage{amsthm}
\usepackage{enumitem}
\usepackage{balance}
\newtheorem{theorem}{Theorem}

\def\BibTeX{{\rm B\kern-.05em{\sc i\kern-.025em b}\kern-.08em
    T\kern-.1667em\lower.7ex\hbox{E}\kern-.125emX}}
\begin{document}

\title{Network Optimization in Dynamic Systems: Fast Adaptation via Zero-Shot Lagrangian Update
}
\author{\IEEEauthorblockN{I-Hong Hou}
\IEEEauthorblockA{\textit{Dept. of ECE} \\
\textit{Texas A$\&$M University}\\
College Station, TX 77843, USA \\
ihou@tamu.edu}
\thanks{
\textcopyright 2024 IEEE. Personal use of this material is permitted. Permission from IEEE must be obtained for all other uses, including reprinting/republishing this material for advertising or promotional purposes, collecting new collected works for resale or redistribution to servers or lists, or reuse of any copyrighted component of this work in other works.
 
This material is based upon work supported in part by NSF under Award Numbers ECCS-2127721 and CCF-2332800, the U.S. Army Research Office under Grant Number W911NF-22-1-015, and the U.S. Office of Naval Research under Grant Number N000142412615.}
}
\maketitle

\begin{abstract}
This paper addresses network optimization in dynamic systems, where factors such as user composition, service requirements, system capacity, and channel conditions can change abruptly and unpredictably. Unlike existing studies that focus primarily on optimizing long-term performance in steady states, we develop online learning algorithms that enable rapid adaptation to sudden changes. Recognizing that many current network optimization algorithms rely on dual methods to iteratively learn optimal Lagrange multipliers, we propose zero-shot updates for these multipliers using only information available at the time of abrupt changes. By combining Taylor series analysis with complementary slackness conditions, we theoretically derive zero-shot updates applicable to various abrupt changes in two distinct network optimization problems. These updates can be integrated with existing algorithms to significantly improve performance during transitory phases in terms of total utility, operational cost, and constraint violations. Simulation results demonstrate that our zero-shot updates substantially improve transitory performance, often achieving near-optimal outcomes without additional learning, even under severe system changes.
\end{abstract}


\section{Introduction}

The study of network optimization has achieved tremendous success in recent decades. Many network utility maximization (NUM) techniques have been proposed and shown to find the optimal control decisions through iterative learning with low complexity and without complete knowledge about the network. However, these techniques typically assume a stationary system. When some aspects of the system change, they need to re-learn the optimal control policy, which can lead to significant performance degradation and constraint violation in the transitory phase. As many modern networked systems, such as mobile wireless networks, can encounter frequent, abrupt, and unpredictable changes in all aspects of the system, including user composition, service requirements, system capacity, and channel conditions, it has become increasingly critical to optimize the system performance not only in the steady state but also during the transitory phases.

To address the above challenges, this paper aims to develop online learning algorithms that can quickly adapt to any changes in the system. We note that many NUM techniques are based on dual methods, where control decisions are characterized by a set of Lagrange multipliers. These techniques then update the Lagrange multipliers in each iteration through, for example, stochastic gradient descent based on empirical observations. Thus, these NUM techniques are fundamentally online learning algorithms for finding the optimal Lagrange multipliers. When some aspects of the system change, the optimal Lagrange multipliers also change. To quickly learn the optimal Lagrange multipliers after some aspects of the system change, we propose a new zero-shot learning algorithm that can obtain an accurate estimate of the new optimal Lagrange multipliers using only information available before the system changes. This zero-shot learning algorithm can then bootstrap existing NUM techniques and improve their performance during transitory phases.

We will study two popular problems in network optimization: distributed rate control for utility maximization and stochastic network optimization with service requirements. The first problem considers a distributed system in which each user chooses its own rate to maximize total utility under a system-wide capacity constraint. A controller aims to induce the optimal rates without knowing the utility function of each individual user. The second problem considers a controller that serves multiple users, each with a strict long-term service requirement, in a stochastic system. The controller aims to satisfy the service requirements of all users while incurring the least amount of operational cost. 

In the distributed rate control for the utility maximization problem, we consider three kinds of abrupt changes. The system capacity may change due to dynamic network slicing or external interference, users may join and leave the system dynamically, and the channel qualities may change. Combining Taylor sequence analysis and complementary slackness conditions, we theoretically derive first-order approximations of the optimal Lagrange multiplier with respect to each of the three kinds of abrupt changes. Furthermore, first-order approximations can be easily calculated using only observable system parameters and past rates of users. We then develop a simple online learning algorithm that employs zero-shot Lagrangian updates based on first-order approximations.

In the stochastic network optimization with service requirements problem, we also consider three kinds of abrupt changes: Users may dynamically adjust their service requirements, users may join and leave the system dynamically, and the mobility patterns of users may change suddenly. Once again, we combine Taylor sequence analysis and complementary slackness conditions to derive first-order approximations of the optimal Lagrange multipliers with respect to each of these three kinds of abrupt changes. We then develop a simple online learning algorithm with zero-shot Lagrangian updates.


In addition to theoretical analysis, we also evaluate the utility of our online learning algorithms with zero-shot updates through comprehensive simulations. The simulation results show that incorporating the zero-shot updates significantly improves the adaptability of NUM techniques, both in terms of the convergence rate of system performance and the accumulated constraint violations. In many settings, the zero-shot updates alone achieve near-optimal performance without any additional iterations.

The rest of the paper is organized as follows: Section~\ref{sec:related} reviews existing work on NUM. Section~\ref{sec:overview} provides an overview of the dual methods in the context of distributed rate control for utility maximization and stochastic network optimization with service requirements. Section~\ref{sec:primal1} considers various system dynamics in the distributed rate control problem and proposes an online learning algorithm with zero-shot updates. Section~\ref{sec:primal2} proposes an onine learning algorithm with zero-shot updates for all possible system dynamics in the stochastic network optimization problem. Section~\ref{sec:simulation} presents the simulation results for the two problems. Finally, Section~\ref{sec:conclusion} concludes the paper.

\section{Related Work} \label{sec:related}

Network optimization has gathered significant research interests. Tassiulas and Ephremides \cite{tassiulas1990stability} have proposed the max-weight scheduling policy and have shown that it is throughput-optimal. Kelly \cite{kelly1997charging} and Kelly et al. \cite{kelly1998rate} have introduced the concept of shadow price and used it to achieve proportional fairness in the network. Subsequent studies \cite{stolyar2005maximizing, lin2006tutorial, palomar2006tutorial, neely2022stochastic, yi2008stochastic} have developed general network utility maximization frameworks. In a nutshell, these studies formulate a network optimization problem as a convex optimization problem, and then show that the evolution of queue status is equivalent to an online gradient algorithm for learning the the Lagrange multipliers or shadow prices. These NUM techniques have been extended to address a multitude of network problems, including packet caching \cite{wang2018distributed}, deadline-constrained networks \cite{hou2010utilitya, hou2010utilityb}, energy-efficient networks \cite{zuo2017energy, jiang2022joint, xia2021online}, microgrid management \cite{zhang2013robust}, and UAV control \cite{wang2018joint, ma2021uav}. In addition, there have been many studies on accelerating the learning rates \cite{liu2016achieving, liu2016heavy, chen2017learn} and dealing with unknown/unobservable system parameters \cite{fu2021learning, verma2020stochastic} of NUM techniques. All these studies assume stationary systems where all parameters of the system are fixed or evolve according to a well-defined stationary random process.

There are a few recent studies on NUM problems in dynamic systems that can change arbitrarily. Liang and Modiano \cite{liang2018network} has studied NUM problems in adversarial environments. Yang et al. \cite{yang2023learning} have proposed using discounted UCB for multi-server systems with time-varying service rates. Nguyen and Modiano \cite{nguyen2023learning} has proposed an online learning algorithm that stabilizes all queues in non-stationary systems asymptotically. However, these studies still focus on long-term performance analysis and ignore transient performance.

\section{An Overview for Dual Methods} \label{sec:overview}

In this section, we provide a brief overview of applying dual methods to two different types of network optimization problems. 

Throughout this paper, we use $\vec{x}$ to denote the vector containing $x_n$ for all $n$. Given a fixed $s$, we use $\vec{x}_s$ to denote the vector that contains $x_{n,s}$ for all $n$. We use $\tilde{X}$ to denote the matrix that contains $x_{n,s}$ for all $n$ and $s$.

\subsection{Distributed Rate Control for Utility Maximization}

Consider a wireless system with an access point (AP) and $N$ wireless users that together operate in a band with $C$ units of bandwidth. The channel quality can be different for different users. We model this heterogeneous channel quality by assuming that it takes $p_n$ units of bandwidth to deliver one unit of data to user $n$. Each user $n$ can dynamically adjust the amount of traffic it generates, which we denote by $x_n$, and receives a utility of $U_n(x_n)$ when it generates $x_n$ traffic. We assume that the utility function $U_n(\cdot)$ is a twice-differentiable, strictly increasing, and strictly concave, in the sense that there exists $\delta>0$ such that $U_n''(x)<-\delta$ for all $x\in[0, C/p_n]$, function. The AP aims to maximize the total utility of the system. Since it takes $p_n$ bandwidth to deliver one unit of data for $n$ and the system only has $C$ units of bandwidth, the AP also needs to ensure $\sum_n p_nx_n\leq C$. Hence, the AP aims to solve the following optimization problem:

\begin{align}
    \mbox{\textbf{Primal1:}\hspace{20pt}}\max_{\vec{x}} & \sum_n U_n(x_n) \label{eq:primal1}\\
    \mbox{s.t. } & \sum_n p_nx_n\leq C,\label{eq:primal2}\\
    &x_n\geq 0, \forall n. \label{eq:primal3}
\end{align}

The \textbf{Primal1} problem is a simple convex optimization problem when all utility functions $U_n(\cdot)$ are known to the AP. However, in practice, users may not wish to reveal their private utility functions. In addition, in many scenarios, users, rather than the AP, determine $x_n$. 

Dual decomposition is a distributed and iterative algorithm to solve the \textbf{Primal1} problem. It introduces a Lagrange multiplier $\lambda$ for the capacity constraint Eq. (\ref{eq:primal2}), which can be interpreted as a shadow price for using one unit of bandwidth. In each round $t$, the AP announces the current shadow price $\lambda_t$ to all users. Since the cost of using one unit of bandwidth is $\lambda_t$ and it takes $p_n$ units of bandwidth to deliver one unit of data for user $n$, user $n$ effectively needs to pay $q_{n,t}:=\lambda_tp_n$ to deliver one unit of data. User $n$ therefore chooses $x_n$ to maximize its own net utility $U_n(x_n)-q_{n,t}x_n$. Specifically, let $g_n(q):=\arg\max_{x:x\geq 0}U_n(x)-qx$, then user $n$ generates $x_{n,t}=g_n(q_{n,t})$ units of data in round $t$. After gathering or observing $x_{n,t}$ for all users, the AP updates the shadow price by $\lambda_{t+1}=\max\{0, \lambda_t+\epsilon_t(\sum_n x_{n,t}-C)\}$, where $\epsilon_t$ is a pre-determined step size, and then proceeds to round $t+1$.

It has been shown that, when $[\epsilon_1, \epsilon_2,\dots]$ is chosen properly, $\lambda_t$ converges to a $\lambda^*$ as $t\rightarrow\infty$. Moreover, $\lambda^*$ is the optimal shadow price in the sense that choosing $x_n=g_n(\lambda^*p_n)$ solves the \textbf{Primal1} problem. We refer interested readers to \cite{lin2004joint, lin2006tutorial, palomar2006tutorial} for a comprehensive introduction and analysis of the dual decomposition method.

Although we use the rate control problem in wireless networks to motivate the formulation of the \textbf{Primal1} problem, this formulation and the dual decomposition method can be applied to many other problems, such as the optimal power flow problem in power systems \cite{mhanna2018component} and big-data processing \cite{notarnicola2017distributed}.

Despite its simplicity and optimality, a critical weakness of the dual decomposition method is that it assumes a stationary system where all aspects of the system, including the set of users, the system capacity $C$, and the channel quality $p_n$ are fixed\footnote{Some extensions of the dual decomposition method consider that the system parameters can be stochastic, but they still require the probability distributions of random variables to be fixed.}. However, in practice, many systems can experience frequent changes. Users may join or leave the system dynamically. The system capacity can change due to, for example, dynamic network slicing. The channel quality can also change frequently due to the mobility of users. Every time any aspect of the system changes, the AP needs to re-learn the optimal shadow price and the system can suffer from low utility and capacity violation during the transitory phase. To address these challenges, we will propose zero-shot algorithms that enable the AP to get an accurate estimate of the optimal shadow price whenever some aspects of the system change.

\subsection{Stochastic Network Optimization with Service Requirements}

Consider a stochastic wireless network with an AP serving $N$ wireless users. Time is slotted and the state of wireless channels can change from slot to slot in an i.i.d. fashion. We use $f_s$ to denote the probability that the channel state is $s\in\mathbb{S}$. At the beginning of each slot $t$, the AP observes the current channel state $s_t$ and then chooses to provide a service rate of $x_{n,t}$ to each user $n$. To provide service rate $\vec{x}_t$, the AP needs to pay a penalty $P_{s_t}(\vec{x}_t)$, where we assume that the state-dependent penalty function $P_s(\cdot)$ is a twice-differentiable, strictly increasing, and strongly convex function for all $s\in\mathbb{S}$. The penalty function can be chosen to reflect, for example, the amount of energy needed to deliver $x_{n,t}$ data to each user $n$. Furthermore, we assume that each user $n$ requires that its long-term service rate, $\lim_{T\rightarrow\infty}\frac{\sum_{t=1}^Tx_{n,t}}{T}$, is at least $A_n$. The AP aims to satisfy the service requirements of all users with minimum cost. Let $\tilde{Y}$ be the matrix containing $y_{n,s}$ for all $n$ and $s$, then the optimal AP policy can be obtained by solving the following optimization problem

\begin{align}
    \mbox{\textbf{Primal2:}\hspace{20pt}}\min_{\tilde{Y}} & \sum_n f_sP_s(\vec{y}_s) \label{eq:primal2-1}\\
    \mbox{s.t. } & \sum_s f_sy_{n,s}\geq A_n, \forall n,\label{eq:primal2-2}\\
    &y_{n,s}\geq 0, \forall n, s, \label{eq:primal2-3}
\end{align}
and then chooses $x_{n,t}=y_{n,s_t}$ in each time slot $t$.

\textbf{Primal2} is a high-dimensional constrained convex optimization problem and may be difficult to solve directly. The celebrated drift-plus-penalty policy \cite{neely2022stochastic} provides a simple and yet near-optimal solution. The drift-plus-penalty policy maintains a virtual queue for each user $n$. The queue length at time $t$, denoted by $Q_{n,t}$, is updated iteratively by
\begin{equation} \label{equation:primal2:q_update}
Q_{n, t+1} = \max\{0, Q_{n, t}+A_n-x_{n,t}\}.
\end{equation}
After observing $s_t$ and $\vec{Q}_t$, the AP chooses $\vec{x}_t$ to minimize $VP_{s_t}(\vec{x}_t)-\sum_n Q_{n,t}x_{n,t}$, where $V>0$ is a constant for the trade-off between penalty and convergence time. It has been shown that the long-term average penalty under the drift-plus-penalty policy is at most $O(1/V)$ worse than the optimal solution to \textbf{Primal2}, but the virtual queue length can be on the order of $\theta(V)$ and, therefore, the system takes $\theta(V)$ time to converge. To address the trade-off between penalty and convergence time, Huang et al. \cite{huang2014power} discovers that the long-term average virtual queue lengths reflect the solution to the dual problem below:
\begin{align}
	&\mbox{\textbf{Dual2:}}\nonumber\\
	&\max_{\vec{\lambda}} \min_{\tilde{Y}} \sum_s f_s\Big(VP_s(\vec{y}_s)-\sum_n \lambda_n(y_{n,s}-A_n)\Big).
\end{align}
Let $\vec{\lambda}^*$ be the optimal solution to \textbf{Dual2}, Huang et al. \cite{huang2014power} proposes choosing $\vec{x}_t$ to minimize $VP_{s_t}(\vec{x}_t)-\sum_n (\lambda^*_n+Q_{n,t})x_{n,t}$ and shows that doing so significantly accelerates the convergence rate of the drift-plus-penalty policy without degrading its penalty performance. 

However, \textbf{Dual2} is a complex convex optimization problem, and the approach of Huang et al. \cite{huang2014power} requires the AP to solve \textbf{Dual2} every time any aspect of the system changes, such as users changing their requirements $A_n$, joining/leaving the system, or changing their mobility patterns, which results in changes in $f_s$. Instead of solving \textbf{Dual2} every time there is a change in the system, we will propose low-complexity zero-shot algorithms to update $\vec{\lambda}^*$.

\section{Online Learning Algorithm with Zero-Shot Updates for Primal1} \label{sec:primal1}

In this section, we discuss how to obtain a zero-shot estimate of the optimal shadow price when some aspects of the system change in the \textbf{Primal1} problem. Specifically, let $\lambda^*$ be the optimal shadow price before the system changes and $\hat{\lambda}^*$ be the optimal shadow price after the system changes. We aim to obtain an accurate estimate of $\hat{\lambda}^*$ using only the information available to the AP before the system changes. We will then leverage the zero-shot estimates to develop an online learning algorithm.

Before we discuss our zero-shot estimation in different kinds of system changes, we first establish some useful properties of \textbf{Primal1}.

Recall that, after receiving a shadow price $\lambda$, each user $n$ chooses its rate $x_{n} = g_n(\lambda p_n)$, where $g_n(q):=\arg\max_{x:x\geq 0}U_n(x)-qx$. We then have 
\begin{equation}
    U'_n(g_n(q))=q, \label{eq:derivative}
\end{equation}
and, by taking derivative on both sides of the above equation,
\begin{equation}
    U''_n(g_n(q))g'_n(q)=1 \label{eq:derivative_inverse}
\end{equation}
Moreover, since $U_n(\cdot)$ is a strictly increasing function for all $n$, we have
\begin{equation}
    \sum_n U_n(g_n(\lambda^*p_n))=C,
\end{equation}
which is usually referred to as the complementary slackness condition.

\subsection{Dynamics in System Capacity}

Consider the situation where the system capacity changes from $C$ to $\hat{C}$. Let $\lambda^*$ be the optimal shadow price when the system capacity is $C$ and let $\hat{\lambda}^*$ be the optimal shadow price when the system capacity is $\hat{C}$. We aim to obtain a good estimate of $\hat{\lambda}^*$ using only information available to the AP at the time of the system capacity change and before making any additional queries to users, that is, the AP estimates $\hat{\lambda}^*$ only based on the values of $C$, $\hat{C}$, $\lambda^*$, $g_n(\lambda^*p_n)$, and $g_n'(\lambda^*p_n)$.

We first establish the following theorem.

\begin{theorem}\label{theorem:primal1:Cchange}
If the system capacity changes from $C$ to $\hat{C}$, then
\begin{equation}
    \hat{\lambda}^*=\lambda^*+\frac{\hat{C}-C}{\sum_n p_n^2 g_n'(\lambda^*p_n)}+O(|\hat{C}-C|^2).
\end{equation}
\end{theorem}
\begin{proof}
    By Eq. (\ref{eq:derivative}), we have $\hat{\lambda}^*p_n=U'_n(g_n(\hat{\lambda}^*p_n))$ and ${\lambda}^*p_n=U'_n(g_n({\lambda}^*p_n))$. Hence, for any $n$,
    \begin{align}
        &(\hat{\lambda}^*p_n-{\lambda}^*p_n)p_ng_n'(\lambda^*p_n)\nonumber\\
        =&[U'_n(g_n(\hat{\lambda}^*p_n))-U'_n(g_n({\lambda}^*p_n))]p_ng_n'(\lambda^*p_n)\nonumber\\
        =&[g_n(\hat{\lambda}^*p_n)-g_n({\lambda}^*p_n)]U''_n(g_n({\lambda}^*p_n))p_ng_n'(\lambda^*p_n) \nonumber\\
        &+ O(|p_ng_n(\hat{\lambda}^*p_n)-p_ng_n({\lambda}^*p_n)|^2) \hspace{20pt}\mbox{(Taylor's Theorem)} \nonumber\\
        =&p_ng_n(\hat{\lambda}^*p_n)-p_ng_n({\lambda}^*p_n)\nonumber\\
        &+ O(|p_ng_n(\hat{\lambda}^*p_n)-p_ng_n({\lambda}^*p_n)|^2),  \label{eq:C_chanage1}
    \end{align}
    where the last step follows Eq. (\ref{eq:derivative_inverse}).

    Since user $n$ chooses $x_n$ to be $g_n(\lambda^*p_n)$ when the shadow price is $\lambda^*$, and $g_n(\hat{\lambda}^*p_n)$ when the shadow price is $\hat{\lambda}^*$, we have the complementary slackness conditions $\sum_n p_ng_n(\lambda^*p_n)=C$ and $\sum_n p_ng_n(\hat{\lambda}^*p_n)=\hat{C}$. Summing Eq. (\ref{eq:C_chanage1}) over all users yield
    \begin{align}
        &\sum_n(\hat{\lambda}^*-{\lambda}^*)p_n^2g_n'(\lambda^*p_n)\nonumber\\
        =&\sum_n p_ng_n(\hat{\lambda}^*p_n)-\sum_np_ng_n({\lambda}^*p_n)\nonumber\\
        &+ \sum_nO(|g_n(\hat{\lambda}^*p_n)-g_n({\lambda}^*p_n)|^2)\nonumber\\
        =&\hat{C}-C+O(|\hat{C}-C|^2), \label{eq:C_chanage2}
    \end{align}
    where $\sum_nO(|p_ng_n(\hat{\lambda}^*p_n)-p_ng_n({\lambda}^*p_n)|^2) \leq O(|\sum_n [p_ng_n(\hat{\lambda}^*p_n)-p_ng_n({\lambda}^*p_n)]|^2)=O(|\hat{C}-C|^2)$ holds because $g_n(\cdot)$ is a decreasing function for all $n$.

    By Eq. (\ref{eq:C_chanage2}), we have $\hat{\lambda}^*=\lambda^*+\frac{\hat{C}-C}{\sum_n p^2_n g_n'(\lambda^*p_n)}+O(|\hat{C}-C|^2)$.
\end{proof}

Using Theorem~\ref{theorem:primal1:Cchange}, we can get a zero-shot estimate of $\hat{\lambda}^*$:
\begin{equation}\label{eq:primal1:estimate Cchange}
    \hat{\lambda}^*\approx\lambda^*+\frac{\hat{C}-C}{\sum_n p_n^2 g_n'(\lambda^*p_n)}.
\end{equation}

\subsection{Dynamics in User Composition}

Consider scenarios in which users may join or leave the system dynamically. We first address the case when a user leaves the system. Without loss of generality, assume that user $N$ leaves the system. Let $\lambda^*$ be the optimal shadow price before user $N$ leaves the system. We aim to obtain a zero-shot estimate of the optimal shadow price after user $N$ leaves the system, denoted by $\hat{\lambda}^*$ using only $C, \lambda^*$, $g_n(\lambda^*p_n)$, and $g_n'(\lambda^*p_n)$.

We can construct an alternative system with only users $1, 2, \dots, N-1$ and system capacity $C-p_Ng_N(\lambda^*p_N)$. It is easy to show that $\lambda^*$ is the optimal shadow price for this alternative system. Thus, from the perspective of users $1, 2, \dots, N-1$, the effect of user $N$ leaving the system is equivalent to having the system capacity in the alternative system increasing from $C-p_Ng_N(\lambda^*p_N)$ to $C$. By Theorem~\ref{theorem:primal1:Cchange}, we can estimate $\hat{\lambda}^*$ by
\begin{equation}\label{eq:primal1:estimate leave}
\hat{\lambda}^*\approx\lambda^*+\frac{p_Ng_N(\lambda^*p_N)}{\sum_{n=1}^{N-1} p_n^2 g_n'(\lambda^*p_n)}.
\end{equation}

Next, we address the case when a new user, numbered as user $N+1$, joins the system. Assuming that we know $p_{N+1}$, $g_{N+1}(\lambda^*p_{N+1})$, and $g_{N+1}'(\lambda^*p_{N+1})$, we can then construct an alternative system with users $1, 2, \dots, N+1$ and system capacity $C+p_{N+1}g_{N+1}(\lambda^*p_{N+1})$. It is easy to show that $\lambda^*$ is the optimal shadow price for this alternative system. Thus, the effect of user $N+1$ joining the system is equivalent to having the system capacity in the alternative system reducing from $C+p_{N+1}g_{N+1}(\lambda^*p_{N+1})$ to $C$. By Theorem~\ref{theorem:primal1:Cchange}, we can estimate $\hat{\lambda}^*$ by
\begin{equation} \label{eq:primal1:estimate join}
\hat{\lambda}^*\approx\lambda^*-\frac{p_{N+1}g_{N+1}(\lambda^*p_{N+1})}{\sum_{n=1}^{N+1} p_n^2 g_n'(\lambda^*p_n)}. 
\end{equation}

In the event that $p_{N+1}$, $g_{N+1}(\lambda^*p_{N+1})$, and $g_{N+1}'(\lambda^*p_{N+1})$ are not immediately known when user $N+1$ joins the system, we can use the average of $p_n$, $g_{n}(\lambda^*p_{n})$, and $g_{n}'(\lambda^*p_{n})$ of all existing users as estimates of $p_{N+1}$, $g_{N+1}(\lambda^*p_{N+1})$, and $g_{N+1}'(\lambda^*p_{N+1})$. We can then apply Eq.~(\ref{eq:primal1:estimate join}) as a zero-shot estimate of $\hat{\lambda}^*$.

\subsection{Dynamics in Channel Qualities}

We now consider that, at a certain point of time, the channel qualities of users change from $\vec{p}$ to $\vec{\hat{p}}$. Let $\lambda^*$ be the optimal shadow price when the channel quality is $\vec{p}$ and let $\hat{\lambda}^*$ be the optimal shadow price when the channel quality is $\vec{\hat{p}}$. The AP aims to obtain a good estimate of $\hat{\lambda}$ using only values of $\vec{p}$, $\vec{\hat{p}}$, $\lambda^*$, $g_n(\lambda^*p_n)$, and $g_n'(\lambda^*p_n)$.

We first establish the following theorem.
\begin{theorem} \label{theorem:primal1:pchanges}
If the channel quality changes from $\vec{p}$ to $\vec{\hat{p}}$, then
    \begin{align}
    \hat{\lambda}^*=& \lambda^* - \sum_n (\hat{p}_n-p_n)\frac{g_n(\lambda^*p_n)+\lambda^*\hat{p}_ng'_n(\lambda^*p_n)}{\sum_m \hat{p}_m^2g'_m(\lambda^*p_m)}\nonumber\\
    &+O(\sum_n|\hat{p}_n-p_n|^2).
    \end{align}
\end{theorem}
\begin{proof}
        By Eq. (\ref{eq:derivative}), we have $\hat{\lambda}^*\hat{p}_n=U'_n(g_n(\hat{\lambda}^*\hat{p}_n))$ and ${\lambda}^*p_n=U'_n(g_n({\lambda}^*p_n))$. Hence, for any $n$,
        \begin{align}
            &\hat{\lambda}^*\hat{p}_n - {\lambda}^*\hat{p}_n \nonumber\\
            =& \hat{\lambda}^*\hat{p}_n - {\lambda}^*{p}_n - \lambda^*(\hat{p}_n-p_n)\nonumber\\
            =&U'_n(g_n(\hat{\lambda}^*\hat{p}_n)) - U'_n(g_n({\lambda}^*p_n)) - \lambda^*(\hat{p}_n-p_n)\nonumber\\
            =&\Big(g_n(\hat{\lambda}^*\hat{p}_n)-g_n({\lambda}^*p_n)\Big)U''_n(g_n({\lambda}^*p_n))- \lambda^*(\hat{p}_n-p_n)\nonumber\\
            &+O(|g_n(\hat{\lambda}^*\hat{p}_n)-g_n({\lambda}^*p_n)|^2) \hspace{30pt}\mbox{(Taylor's Theorem)}\nonumber\\
            =&\Big(g_n(\hat{\lambda}^*\hat{p}_n)-g_n({\lambda}^*p_n)\Big)/g'_n({\lambda}^*p_n)- \lambda^*(\hat{p}_n-p_n)\nonumber\\
            &+O(|g_n(\hat{\lambda}^*\hat{p}_n)-g_n({\lambda}^*p_n)|^2) \hspace{30pt}\mbox{($\because$ Eq.~(\ref{eq:derivative_inverse})).} \label{eq:p_change1}
        \end{align}

        Since $U''_n(g_n(q))< -\delta < 0$ for all $q$, we have $|g'_n(q)|=|1/U''_n(g_n(q))|\leq 1/\delta$ for all $q$. Hence, $|g_n(\hat{\lambda}^*\hat{p}_n)-g_n({\lambda}^*p_n)|\leq |\hat{\lambda}^*\hat{p}_n-{\lambda}^*p_n|/\delta$ and 
        \begin{equation}
            O(|g_n(\hat{\lambda}^*\hat{p}_n)-g_n({\lambda}^*p_n)|^2)= O(|\hat{\lambda}^*\hat{p}_n-{\lambda}^*p_n|^2). \label{eq:p_change2}
        \end{equation}

        Moreover, due to the complementary slackness conditions $\sum_n p_ng_n(\lambda^*p_n)=\sum_n \hat{p}_ng_n(\hat{\lambda}^*\hat{p}_n)=C$, we have 
        \begin{align}
            &\sum_n \hat{p}_n\Big(g_n(\hat{\lambda}^*\hat{p}_n)-g_n({\lambda}^*p_n)\Big)\nonumber\\
            =&\sum_n \hat{p}_ng_n(\hat{\lambda}^*\hat{p}_n) - \sum_n {p}_ng_n({\lambda}^*{p}_n) + \sum_n (p_n-\hat{p}_n)g_n({\lambda}^*{p}_n)\nonumber\\
            =&\sum_n (p_n-\hat{p}_n)g_n({\lambda}^*{p}_n). \label{eq:p_change3}
        \end{align}

        Combining Eq. (\ref{eq:p_change1}), (\ref{eq:p_change2}), and (\ref{eq:p_change3}) and we have
        \begin{align}
            &\sum_n (\hat{\lambda}^*\hat{p}_n - {\lambda}^*\hat{p}_n)g'_n(\lambda^*p_n)\hat{p}_n\nonumber\\
            =&-\sum_n (\hat{p}_n-{p}_n)\Big(g_n(\lambda^*p_n)+\lambda^*\hat{p}_ng'_n(\lambda^*p_n)\Big)\nonumber\\
            &+O(\sum_n|\hat{\lambda}^*\hat{p}_n-\lambda^*p_n|^2) \label{eq:p_change4}
        \end{align}
        and hence
            \begin{align}
    \hat{\lambda}^*-\lambda^*=&  - \sum_n (\hat{p}_n-p_n)\frac{g_n(\lambda^*p_n)+\lambda^*\hat{p}_ng'_n(\lambda^*p_n)}{\sum_m \hat{p}_m^2g'_m(\lambda^*p_m)}\nonumber\\
    &+O(\sum_n|\hat{\lambda}^*\hat{p}_n-\lambda^*p_n|^2). \label{eq:p_change5}
    \end{align}

        Finally, note that $O(\sum_n|\hat{\lambda}^*\hat{p}_n-\lambda^*p_n|^2)=O(\sum_n|(\hat{\lambda}^*-\lambda^*)\hat{p}_n-\lambda^*(p_n-\hat{p}_n|^2)=O(|\hat{\lambda}^*-\lambda^*|^2)+O(|(\hat{\lambda}^*-\lambda^*)\sum_n(\hat{p}_n-p_n)|)+O(\sum_n|\hat{p}_n-p_n|^2)$. Apply Eq.~(\ref{eq:p_change5}) again and we have $O(\sum_n|\hat{\lambda}^*\hat{p}_n-\lambda^*p_n|^2)=O(\sum_n|\hat{p}_n-p_n|^2),$
        which completes the proof.
\end{proof}

Using Theorem~\ref{theorem:primal1:pchanges}, we can get a zero-shot estimate of $\hat{\lambda}^*$:
\begin{equation}\label{eq:primal1:estimate pchange}
    \hat{\lambda}^*\approx \lambda^* - \sum_n (\hat{p}_n-p_n)\frac{g_n(\lambda^*p_n)+\lambda^*\hat{p}_ng'_n(\lambda^*p_n)}{\sum_m \hat{p}_m^2g'_m(\lambda^*p_m)}.
\end{equation}

\subsection{Online Learning Algorithm for Fast Adaptation}

Equipped with zero-shot estimates in the previous three sections, we now present a simple online learning algorithm that can quickly adapt to changes in \textbf{Primal1}.

Our zero-shot estimates, Eq. (\ref{eq:primal1:estimate Cchange}), (\ref{eq:primal1:estimate leave}), (\ref{eq:primal1:estimate join}), and (\ref{eq:primal1:estimate pchange}), only involve $\lambda^*$, $g_n(\lambda^*p_n)$, $g'_n(\lambda^*p_n)$, $C$, and $\vec{p}$. $C$ and $\vec{p}$ are directly observable by the AP, so our algorithm only needs to learn $\lambda^*$, $g_n(\lambda^*p_n)$, and $g'_n(\lambda^*p_n)$. When there is no change in the system, we update the shadow price iteratively by $\lambda_{t+1}=\max\{0, \lambda_t+\epsilon_t(\sum_n x_{n,t}-C)\}$ in each round. Lin, Shroff, and Srikant \cite{lin2006tutorial} has shown that this iterative update ensures $\lim_{t\rightarrow\infty}\lambda_t=\lambda^*$ when the step size $\epsilon_t$ is properly chosen. Suppose some aspects of the system change at the end of round $T$. Since $\lim_{t\rightarrow\infty}\lambda_t=\lambda^*$, we estimate $\lambda^*$ by $\lambda_T$. Since $x_{n,t}=g_n(\lambda_tp_n)$, we estimate $g_n(\lambda^*p_n)$ by $x_{n,T}$. Finally, since $g'_n(\lambda^*p_n)=\lim_{\Delta\rightarrow 0}\frac{g_n((\lambda^*+\Delta)p_n)-g_n(\lambda^*p_n)}{p_n\Delta}$, we estimate $g'_n(\lambda^*p_n)$ by $\frac{x_{n,T}-x_{n,T-1}}{(\lambda_T-\lambda_{T-1})p_n}$. Thus, the AP can easily estimate $\lambda^*$, $g_n(\lambda^*p_n)$, and $g'_n(\lambda^*p_n)$ using only information that is directly available to it with minimum computation. We summarize our algorithm in Alg.~\ref{alg:primal1}.

\begin{algorithm}[h]
   \caption{Online Learning with Zero-Shot Lagrangian Updates for \textbf{Primal1}}
   \label{alg:primal1}
\begin{algorithmic}[1]
\STATE $\lambda_1\leftarrow 0$
\FOR{each iteration $t=1,2,\dots$}
\STATE Announce $\lambda_t$ to all users
\STATE Collect $x_{n,t}$ from all users $n$
\STATE $\lambda_{t+1}\leftarrow\max\{0, \lambda_t+\epsilon_t(\sum_n x_{n,t}-C)\}$
\IF{Some aspects of the system change}
\STATE $\lambda^*\leftarrow \lambda_{t}$
\STATE $g_n(\lambda^*p_n)\leftarrow x_{n,t},$ for all $n$
\STATE $g'_n(\lambda^*p_n)\leftarrow\frac{x_{n,t}-x_{n,t-1}}{(\lambda_t-\lambda_{t-1})p_n}$, for all $n$
\IF{capacity changes from $C$ to $\hat{C}$}
\STATE Set $\lambda_{t+1}$ by Eq.~(\ref{eq:primal1:estimate Cchange})
\ENDIF
\IF{a user leaves the system}
\STATE Set $\lambda_{t+1}$ by Eq.~(\ref{eq:primal1:estimate leave})
\ENDIF
\IF{a user joins the system}
\STATE Set $\lambda_{t+1}$ by Eq.~(\ref{eq:primal1:estimate join})
\ENDIF
\IF{channel quality changes from $\vec{p}$ to $\vec{\hat{p}}$}
\STATE Set $\lambda_{t+1}$ by Eq.~(\ref{eq:primal1:estimate pchange})
\ENDIF
\ENDIF
\ENDFOR
\end{algorithmic}
\end{algorithm}

\section{Online Learning Algorithm with Zero-Shot Updates for Primal2} \label{sec:primal2}

In this section, we study the stochastic network optimization problem with service requirements as described in \textbf{Primal2}. We discuss how to get a zero-shot estimate of the optimal solution to the \textbf{Dual2} problem when some aspects of the system change. Specifically, let $\vec{\lambda}^*$ and $\vec{\hat{\lambda}}^*$ be the optimal solutions to the \textbf{Dual2} problem before and after the system changes, respectively. We aim to get an accurate estimate of $\vec{\hat{\lambda}}^*$ using only information available to the AP before the system changes.

We first establish some useful properties for \textbf{Primal2} and \textbf{Dual2}.

Let $\tilde{G}(\vec{\lambda})$ be the matrix containing all functions $g_{n,s}(\vec{\lambda})$ such that choosing $\tilde{Y}=\tilde{G}(\vec{\lambda})$ minimizes $\sum_s f_s\Big(VP_s(\vec{y}_s)-\sum_n \lambda_n(y_{n,s}-A_n)\Big)$. Due to the summation form of $\sum_s f_s\Big(VP_s(\vec{y}_s)-\sum_n \lambda_n(y_{n,s}-A_n)\Big)$, we have $\vec{g}_s(\vec{\lambda})=\arg\min_{\vec{y}_s}VP_s(\vec{y}_s)-\sum_n \lambda_n(y_{n,s}-A_n)$. 

Since $\vec{g}_s(\vec{\lambda})=\arg\min_{\vec{y}_s}VP_s(\vec{y}_s)-\sum_n \lambda_n(y_{n,s}-A_n)$, we have
\begin{equation}
\frac{\partial}{\partial y_{n,s}}P_s(\vec{y}_s)\Big|_{\vec{y}_s=\vec{g}_s(\vec{\lambda})}=\lambda_n/V, \forall n,s.\label{equation:primal2derivative}
\end{equation}
Taking the partial derivative with respect to $\lambda_m$ on both sides of the above equation yields
\begin{align}
&\sum_k\frac{\partial^2}{\partial y_{n,s}\partial y_{k,s}}P_s(\vec{y}_s)\Big|_{\vec{y}_s=\vec{g}_s(\vec{\lambda})} \frac{\partial g_{k,s}(\vec{\lambda}^*)}{\partial \lambda_m}\nonumber\\
=&\begin{cases}
1/V, &\text{if $m=n$},\\
0, &\text{otherwise.}
\end{cases}
\label{equation:primal2hessian}
\end{align}

Since $P_s(\vec{y}_s)$ is strongly convex, there exists $\delta>0$ such that $\vec{z}^T\nabla^2P_s(\vec{y}_s)\vec{z}\geq \delta \parallel \vec{z}\parallel^2$ for any vector $\vec{z}$. For fixed $m$ and $s$, consider the vector that contains $\frac{\partial g_{k,s}(\vec{\lambda}^*)}{\partial \lambda_m}$ for all $k$. We then have
\begin{align}
    &\delta \sum_k |\frac{\partial g_{k,s}(\vec{\lambda}^*)}{\partial \lambda_m}|^2\nonumber\\
    \leq& \sum_n\sum_k \frac{\partial g_{n,s}(\vec{\lambda}^*)}{\partial \lambda_m}\frac{\partial^2}{\partial y_{n,s}\partial y_{k,s}}P_s(\vec{y}_s)\Big|_{\vec{y}_s=\vec{g}_s(\vec{\lambda})} \frac{\partial g_{k,s}(\vec{\lambda}^*)}{\partial \lambda_m}\nonumber\\
    =&\frac{1}{V}\frac{\partial g_{m,s}(\vec{\lambda}^*)}{\partial \lambda_m}\leq \frac{1}{V}\sum_k |\frac{\partial g_{k,s}(\vec{\lambda}^*)}{\partial \lambda_m}|.
\end{align}
Using Cauchy–Schwarz inequality, we can then show that 
\begin{equation}
    |\frac{\partial g_{k,s}(\vec{\lambda}^*)}{\partial \lambda_m}|\leq \frac{N}{\delta V}, \forall m,k,s. \label{equation:primal2g_bound}
\end{equation}

Moreover, since $P_s(\cdot)$ is an increasing function for each $s$, the complementary slackness conditions can be written as 
\begin{equation}
\sum_s f_sg_{n,s}(\lambda^*)=A_n,
\end{equation}
for all $n$.

\subsection{Dynamics in Service Requirements}

Recall that each user $n$ demands a long-term service requirement of $A_n$. We now address the situation where users change their service requirements from $\vec{A}$ to $\vec{\hat{A}}$.

Let $h_{n,m}:=\sum_sf_s\frac{\partial g_{m,s}(\vec{\lambda}^*)}{\partial \lambda_n}$ and let $\tilde{H}$ be the $N\times N$ matrix containing all $h_{n,m}$. We first establish the following theorem:

\begin{theorem} \label{theorem:primal2:Achange}
If the service requirements change from $\vec{A}$ to $\vec{\hat{A}}$, then 
\begin{equation}
\vec{\hat{\lambda}}^*= \vec{\lambda}^*+\tilde{H}^{-1}(\vec{\hat{A}}-\vec{A})+O(\parallel \vec{\hat{A}}-\vec{A}\parallel^2).
\end{equation}
\end{theorem}
\begin{proof}
For a fixed $n$, we have
\begin{align}
&\sum_m h_{n,m}(\hat{\lambda}^*_m-\lambda^*_m)\nonumber\\
=&V\sum_m \sum_sf_s\frac{\partial g_{m,s}(\vec{\lambda}^*)}{\partial \lambda_n}
\Big( \frac{\partial}{\partial y_{m,s}}P_s(\vec{y}_s)\Big|_{\vec{y}_s=\vec{g}_s(\vec{\hat{\lambda}}^*)}\nonumber\\
& -  \frac{\partial}{\partial y_{m,s}}P_s(\vec{y}_s)\Big|_{\vec{y}_s=\vec{g}_s(\vec{\lambda}^*)} \Big) \hspace{60pt}\text{($\because$ Eq. (\ref{equation:primal2derivative}))} \nonumber\\
=&V\sum_m \sum_k \sum_sf_s\frac{\partial g_{m,s}(\vec{\lambda}^*)}{\partial \lambda_n} \Big[\big(g_{k,s}(\vec{\hat{\lambda}}^*)-g_{k,s}(\vec{\lambda}^*)\big)\times\nonumber\\
&\frac{\partial^2}{\partial y_{m,s}\partial y_{k,s}}P_s(\vec{y}_s)\Big|_{\vec{y}_s=\vec{g}_s(\vec{\lambda})}\Big]\nonumber\\
&+O(\sum_k\sum_s\big(g_{k,s}(\vec{\hat{\lambda}}^*)-g_{k,s}(\vec{\lambda}^*)\big)^2) \hspace{15pt}\text{(Taylor's Theorem)}\nonumber\\
=& \sum_s f_s g_{n,s}(\vec{\hat{\lambda}}^*) - \sum_s f_sg_{n,s}(\vec{\lambda}^*)\nonumber\\
&+O(\sum_k\sum_s\big(g_{k,s}(\vec{\hat{\lambda}}^*)-g_{k,s}(\vec{\lambda}^*)\big)^2) \hspace{15pt}\text{($\because$ Eq. (\ref{equation:primal2hessian}))}\nonumber\\
=& \hat{A}_n-A_n+O(\sum_k\sum_s\big(g_{k,s}(\vec{\hat{\lambda}}^*)-g_{k,s}(\vec{\lambda}^*)\big)^2). \label{eq:primal2:Achange1}
\end{align}

Combining the above equation over all users $m$, we have
\begin{equation} \label{eq:primal2:Achange2}
\vec{\hat{\lambda}}^*-\vec{\lambda}^*=\tilde{H}^{-1}(\vec{\hat{A}}-\vec{A})+O(\sum_{k,s}\big(g_{k,s}(\vec{\hat{\lambda}}^*)-g_{k,s}(\vec{\lambda}^*)\big)^2).
\end{equation}
Finally, due to Eq. (\ref{equation:primal2g_bound}), we have $\big(g_{k,s}(\vec{\hat{\lambda}}^*)-g_{k,s}(\vec{\lambda}^*)\big)^2=O(\parallel \vec{\hat{\lambda}}^*-\vec{\lambda}^*\parallel^2)$. Apply Eq. (\ref{eq:primal2:Achange2}) again completes the proof.
\end{proof}

Using Theorem~\ref{theorem:primal2:Achange}, we can obtain a zero-shot estimate of $\vec{\hat{\lambda}}^*$:
\begin{equation}    \label{equation:primal2:Achange_update}
    \vec{\hat{\lambda}}^*\approx \vec{\lambda}^*+\tilde{H}^{-1}(\vec{\hat{A}}-\vec{A})
\end{equation}

\subsection{Dynamics in User Composition}

Consider the scenario in which users may join or leave the system dynamically. We first address the case when a user leaves the system. Without loss of generality, assume that user $N$ leaves the system. User $N$ leaving the system is equivalent to the scenario where user $N$ reduces its service requirement to 0. Hence, we can construct $\vec{\hat{A}}$ such that $\hat{A}_N=0$ and $\hat{A}_n=A_n$ for all $n\neq N$. We then apply Eq.~(\ref{equation:primal2:Achange_update}) to obtain a zero-shot estimate of $\vec{\hat{\lambda}}^*$.

Next, we address the case when a new user, numbered as user $N+1$, joins the system. Obtaining a zero-shot estimate of $\vec{\hat{\lambda}}^*$ is harder in this case because we have little information about user $N+1$ other than its service requirement $A_{N+1}$. To obtain a zero-shot estimate, we first construct $\vec{\hat{A}}$ such that $\hat{A}_n=A_n+\frac{A_{N+1}}{N}$. This is to make $\sum_{n=1}^{N}\hat{A}_n=\sum_{n=1}^{N+1}A_{n}$. We then apply Eq.~(\ref{equation:primal2:Achange_update}) to obtain a zero-shot estimate of $\hat{\lambda}^*_1, \hat{\lambda}^*_2, \dots$. Finally, we estimate $\hat{\lambda}^*_{N+1}$ as $\sum_{n=1}^N\hat{\lambda}^*_n/N$.

\subsection{Dynamics in Channel State Distribution}

Recall that $f_s$ is the probability that the channel state at particular point of time is $t$. In wireless communications, the channel quality of a link is mainly influenced by the mobility pattern of the user. When mobility patterns change, $f_s$ also changes. In this section, we discuss the scenario when the channel state distribution changes from $\vec{f}$ to $\vec{\hat{f}}$.

Let $\vec{\lambda}^*$ and $\vec{\hat{\lambda}}^*$ be the optimal solution to the \textbf{Dual2} problem before and after the channel state distribution changes, respectively. We seek to obtain a zero-shot estimate of $\vec{\hat{\lambda}}^*$.

Let $\hat{h}_{n,m}:=\sum_s\hat{f}_s\frac{\partial g_{m,s}(\vec{\lambda}^*)}{\partial \lambda_n}$ and let $\tilde{\hat{H}}$ be the $N\times N$ matrix containing all $\hat{h}_{n,m}$. Let $d_n:=\sum_s(f_s-\hat{f}_s)g_{n,s}(\vec{\lambda})$ and let $\vec{d}$ be the vector containing all $d_n$. We have the following theorem:

\begin{theorem} \label{theorem:primal2:Fchange}
If the channel state distribution changes from $\vec{f}$ to $\vec{\hat{f}}$, then
\begin{equation}
\vec{\hat{\lambda}}^*= \vec{\lambda}^*+\tilde{\hat{H}}^{-1}\vec{d}+O(\sum_s (f_s-\hat{f}_s)^2).
\end{equation}
\end{theorem}
\begin{proof}
The proof is similar to that for Theorem~\ref{theorem:primal2:Achange}, and is hence omitted.

\end{proof}
Using Theorem~\ref{theorem:primal2:Fchange}, we can obtain a zero-shot estimate of $\vec{\hat{\lambda}}^*$:
\begin{equation} \label{equation:primal2:fchange}
\vec{\hat{\lambda}}^*\approx \vec{\lambda}^*+\tilde{\hat{H}}^{-1}\vec{d}.
\end{equation}

\begin{algorithm}[h]
   \caption{Online Learning with Zero-Shot Lagrangian Updates in \textbf{Primal2}}
   \label{alg:primal2}
\begin{algorithmic}[1]
\STATE Initialize $\vec{\lambda}^0$
\STATE $Q_{n,1}\rightarrow 0, \forall n$
\FOR{each time slot $t=1,2,\dots$}
\STATE Observe $s_t$
\STATE Choose $\vec{x}_t$ to minimize $VP_{s_t}(\vec{x}_t)-\sum_n (\lambda^0_n+Q_{n,t})x_{n,t}$
\STATE $Q_{n,t+1}\leftarrow\max\{0, Q_{n,t}+A_n-x_{n,t}\}, \forall n$
\IF{Some aspects of the system change}
\STATE $\lambda_n^*\leftarrow \lambda_n^0+Q_{n,t}$
\FOR{each state $s$}
\STATE $\vec{g}_s(\vec{\lambda}^*)\leftarrow\arg\min_{\vec{y}_s}VP_s(\vec{y}_s)-\sum_n \lambda^*_n(y_{n,s}-A_n)$
\STATE $\nabla\vec{g}_s(\vec{\lambda}^*)\leftarrow\frac{1}{V}\Big(\nabla^2P_s(\vec{y}_s)\Big|_{\vec{y}_s=\vec{g}_s(\vec{\lambda}^*)}\Big)^{-1}$
\ENDFOR
\IF{system requirements change or users join/leave the system}
\STATE $h_{n,m}\leftarrow \sum_s f_s \frac{\partial g_{m,s}(\vec{\lambda}^*)}{\partial \lambda_n}, \forall n,m$
\STATE Calculate $\vec{\hat{\lambda}}^*$ by (\ref{equation:primal2:Achange_update})
\ENDIF
\IF{channel state distribution changes}
\STATE $\hat{h}_{n,m}\leftarrow \sum_s \hat{f}_s \frac{\partial g_{m,s}(\vec{\lambda}^*)}{\partial \lambda_n}, \forall n,m$
\STATE Calculate $\vec{\hat{\lambda}}^*$ by (\ref{equation:primal2:fchange})
\ENDIF
\STATE $\vec{\lambda}^0\leftarrow\vec{\lambda}^0+(\vec{\hat{\lambda}}^*-\vec{\lambda}^*)$
\ENDIF
\ENDFOR
\end{algorithmic}
\end{algorithm}

\subsection{Online Learning Algorithm for Fast Adaptation}

This section presents a simple online learning algorithm that can quickly adapt to changes in \textbf{Primal2}.

We first discuss the estimate of $\vec{\lambda}^*$. When the system starts, the AP initializes a vector $\vec{\lambda}^0$ as its initial guess of $\vec{\lambda}^*$. The AP can set $\vec{\lambda}^0=\vec{\lambda}^*$ if it can solve \textbf{Dual2} or set $\vec{\lambda}^0$ to be an arbitrary constant vector if solving \textbf{Dual2} is computationally infeasible. The AP employs the drift-plus-penalty algorithm to determine service rates. In each time slot $t$, it chooses $\vec{x}_t$ to minimize $VP_{s_t}(\vec{x}_t)-\sum_n (\lambda^0_n+Q_{n,t})x_{n,t}$ and then updates $Q_{n,t+1}$ using Eq. (\ref{equation:primal2:q_update}).

Suppose some aspects of the system change at the end of time slot $T$ and the AP needs to obtain a zero-shot estimate of $\vec{\hat{\lambda}}^*$. Our zero-shot estimates discussed in the previous sections only require the AP to know $\vec{\lambda}^*$, $g_{n,s}(\vec{\lambda}^*)$ and $\frac{\partial g_{m,s}(\vec{\lambda}^*)}{\partial \lambda^*_n}$ for all $n, m,$ and $k$ to calculate the zero-shot estimates. Thus, we fist discuss how to estimate $\vec{\lambda}^*$, $g_{n,s}(\vec{\lambda}^*)$ and $\frac{\partial g_{m,s}(\vec{\lambda}^*)}{\partial \lambda^*_n}$.

Using a similar argument of Huang et al. \cite{huang2014power}, one can show that $\lambda^0_n+Q_{n,t}$ plays the role of $\lambda^*_n$. Hence, we will estimate $\lambda^*_n$ by $\lambda^0_n+Q_{n,T}$. Since the AP knows the penalty function $P_s(\cdot)$ for all $s$, it is then straightforward for the AP to calculate $g_{n,s}(\vec{\lambda}^*)$ using the estimated $\lambda^*_n$. Finally, to calculate $\frac{\partial g_{m,s}(\vec{\lambda}^*)}{\partial \lambda^*_n}$, we let $\nabla\vec{g}_s(\vec{\lambda}^*)$ be the matrix containing $\frac{\partial g_{m,s}(\vec{\lambda}^*)}{\partial \lambda^*_n}$ for all $m$ and $n$. Eq.~(\ref{equation:primal2hessian}) shows that $\nabla\vec{g}_s(\vec{\lambda}^*)=\frac{1}{V}\Big(\nabla^2P_s(\vec{y}_s)\Big|_{\vec{y}_s=\vec{g}_s(\vec{\lambda}^*)}\Big)^{-1}$, which can be easily calculated by the AP.

After obtaining $\vec{\lambda}^*$, $g_{n,s}(\vec{\lambda}^*)$ and $\frac{\partial g_{m,s}(\vec{\lambda}^*)}{\partial \lambda^*_n}$, the AP employs the techniques in the previous sections to calculate a zero-shot estimate of $\vec{\hat{\lambda}}^*$. The AP then sets $\vec{\lambda}^0=\vec{\lambda}^0+(\vec{\hat{\lambda}}^*-\vec{\lambda}^*)$, so as to ensure that $\lambda^0_n+Q_{n,T}$ becomes an accurate estimate of $\hat{\lambda}^*_n$. We summarize the complete algorithm in Alg.~\ref{alg:primal2}.

\section{Simulation Results} \label{sec:simulation}
\begin{figure*}[t]
    \begin{center}
    \subfigure[Dynamics in system capacity]
    {
       \includegraphics[width=0.3\linewidth]{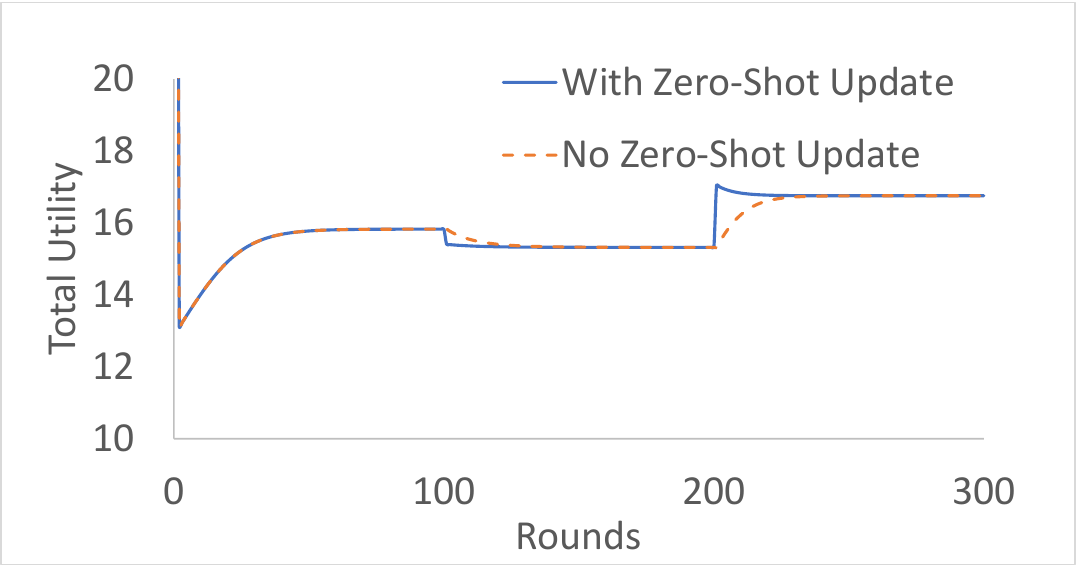}
       \label{fig:primal1_c_utility}
    }
    \subfigure[Dynamics in user composition]
    {
       \includegraphics[width=0.3\linewidth]{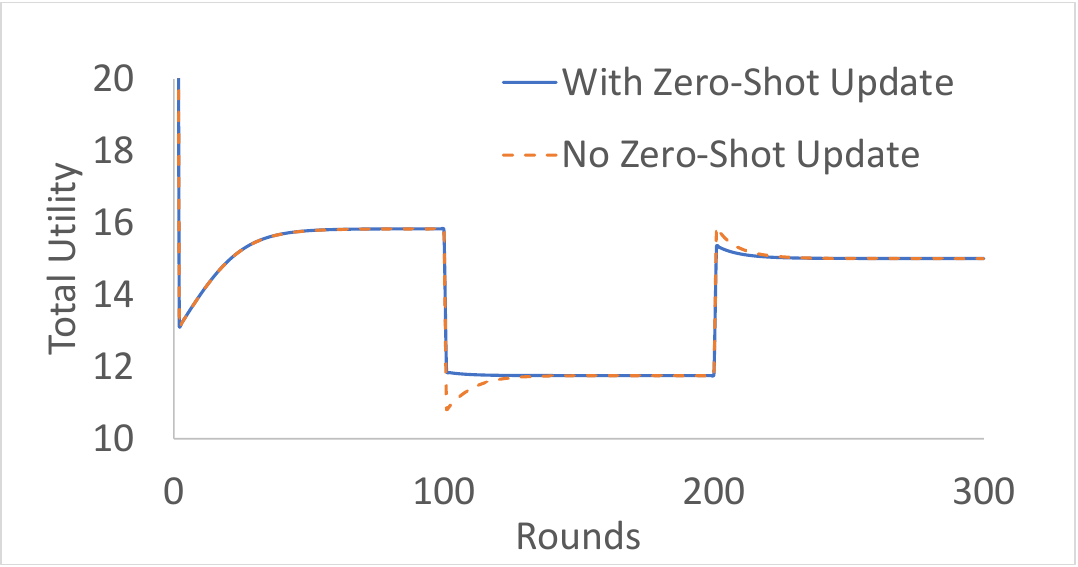}
       \label{fig:primal1_n_utility}
    }
    \subfigure[Dynamics in channel qualities]
    {
       \includegraphics[width=0.3\linewidth]{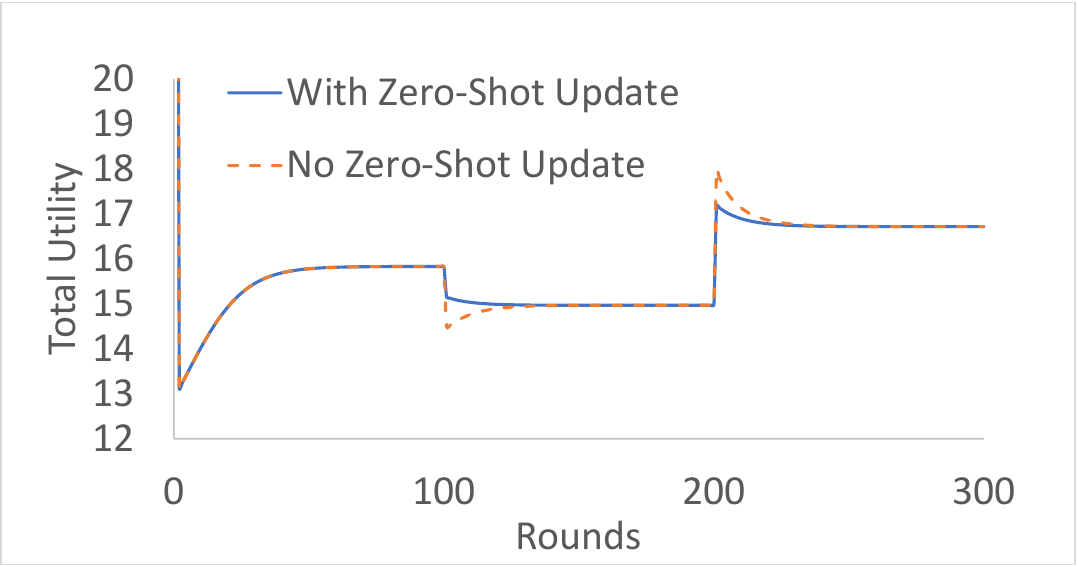}
       \label{fig:primal1_p_utility}
    }
    \end{center}
    \caption{Total utility for \textbf{Primal1} under various system dynamics.}
    \label{fig:primal1_utility}
\end{figure*}

\begin{figure*}[t]
    \begin{center}
    \subfigure[Dynamics in system capacity]
    {
       \includegraphics[width=0.3\linewidth]{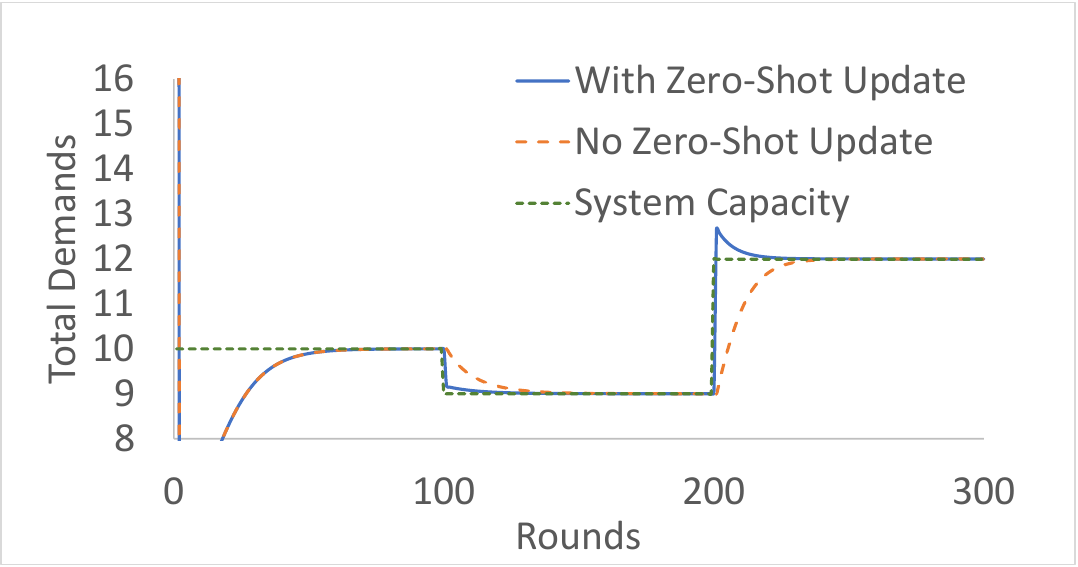}
       \label{fig:primal1_c_usage}
    }
    \subfigure[Dynamics in user composition]
    {
       \includegraphics[width=0.3\linewidth]{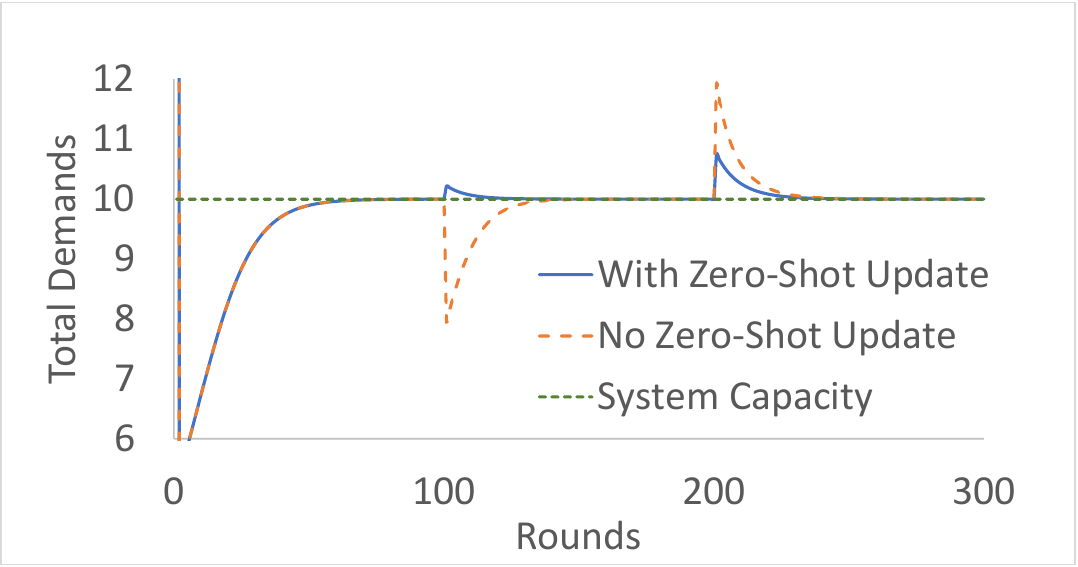}
       \label{fig:primal1_n_usage}
    }
    \subfigure[Dynamics in channel qualities]
    {
       \includegraphics[width=0.3\linewidth]{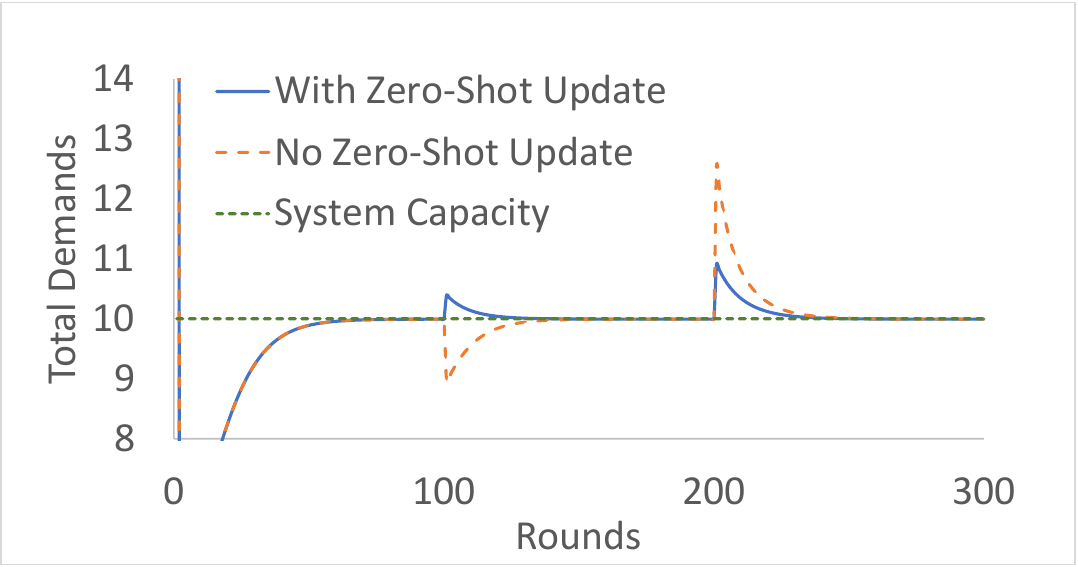}
       \label{fig:primal1_p_usage}
    }
    \end{center}
    \caption{Total resource demands for \textbf{Primal1} under various system dynamics.}
    \label{fig:primal1_utility}
\end{figure*}

\subsection{Distributed Rate Control for Utility Maximization}

We consider a system with capacity $C$ and $N$ users. Each user $n$ has a utility function $U_n(x)$ and channel quality $p_n$. When the system starts at time 0, $C = 10$ and $N=11$. The utility function of user $n$ is $U_n(x)=\frac{x^{\gamma_n}}{\gamma_n}$, where $\gamma_n=0.14+0.06n$, and the channel quality of user $n$ is $p_n=1.4+0.6n$. We will evaluate the performance of our zero-shot update, as shown in Alg.~\ref{alg:primal1}, in three scenarios of system dynamics: dynamics in the system capacity, dynamics in user composition, and dynamics in channel qualities. For each scenario, we evaluate the total utility of the system, $\sum_n U_n(x_n)$, and the total resource demands of all users, $\sum_n p_nx_n$. We note that the system converges to the optimal solution when $\sum_n p_nx_n$ converges to the capacity $C$. We set $\epsilon_t=0.003$ in Alg.~\ref{alg:primal1}. The performance of our zero-shot update is compared against a baseline policy without zero-shot update that only executes lines 1 -- 5 in Alg.~\ref{alg:primal1}.\\

\noindent\textbf{Dynamics in system capacity.} In this scenario, we decrease the system capacity from $C=10$ to $C=9$ at the end of the 100-th round and then increase it from $C=9$ to $C=12$ at the end of the 200-th round. The simulation results are shown in Figs. \ref{fig:primal1_c_utility} and \ref{fig:primal1_c_usage}. When the system capacity decreases from 10 to 9, representing a $10\%$ change, our algorithm with zero-shot update converges to the optimal solution almost immediately. When the system capacity increases from 9 to 12, representing a $33\%$ change, our algorithm with zero-shot update takes a few rounds to converge, but it still converges much faster than the algorithm without zero-shot update. To better understand how the amount of capacity change impacts the convergence speed, we simulate the scenario when the system capacity increases by $z\%$, where $z$ ranges from 5 to 95, and evaluate the number of additional rounds needed for the total resource demand to converge to within $1\%$ from the system capacity. Simulation results are shown in Fig. \ref{fig:primal1_c_convtime}. It can be observed that applying the zero-shot update significantly increases convergence speed. When $z$ is as large as $20$, the zero-shot update alone is able to make $\sum_n p_nx_n$ within $1\%$ from $C$ without any additional updates.\\

\begin{figure}[t]
    \begin{center}
       \includegraphics[width=0.7\linewidth]{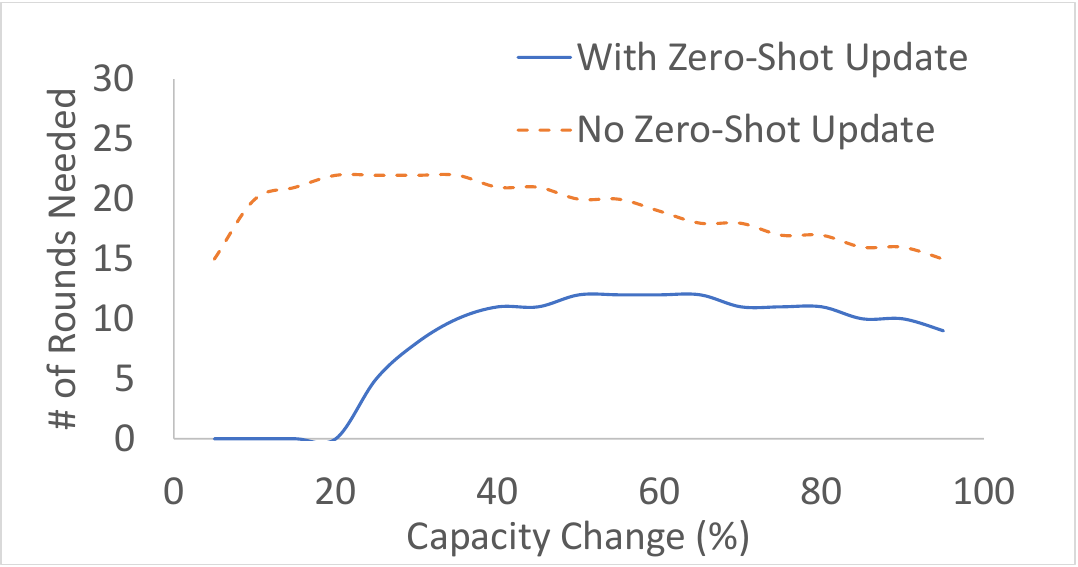}
   \end{center}
    \caption{Convergence speed when the system capacity changes.}
    \label{fig:primal1_c_convtime}
\end{figure}

\noindent\textbf{Dynamics in user composition.} In this scenario, user 1 leaves the system at the end of the 100-th round. Another user, numbered as user 12, with $\gamma_{12}=0.3$ and $p_{12}=5$ joins the system at the end of the 200-th round. We assume that the AP does not know $\gamma_{12}$ and $p_{12}$ when user 12 joins. Hence, it uses the average of $p_n$, $g_{n}(\lambda^*p_{n})$, and $g_{n}'(\lambda^*p_{n})$ of all existing users as estimates to those of user 12 before applying the zero-shot update. Simulation results are shown in Figs. \ref{fig:primal1_n_utility} and \ref{fig:primal1_n_usage}. It can be observed that the zero-shot update significantly improves the convergence speed. We also note that the zero-shot update is more accurate when user 1 leaves the system and is less accurate when user 12 joins the system. This is because we assume no knowledge on the values of $p_n$, $g_{n}(\lambda^*p_{n})$, and $g_{n}'(\lambda^*p_{n})$ of user 12.\\

\noindent\textbf{Dynamics in channel qualities.} In this scenario, all $p_n$ increase by $20\%$ at the end of the 100-th round and decrease by $30\%$ at the end of the 200-th round. Simulation results, as shown in Figs. \ref{fig:primal1_p_utility} and \ref{fig:primal1_p_usage}, show that applying the zero-shot update significantly improves the convergence speed.

\subsection{Stochastic Network Optimization with Service Requirements}

\begin{figure*}[t]
    \begin{center}
    \subfigure[Dynamics in service requirements]
    {
       \includegraphics[width=0.3\linewidth]{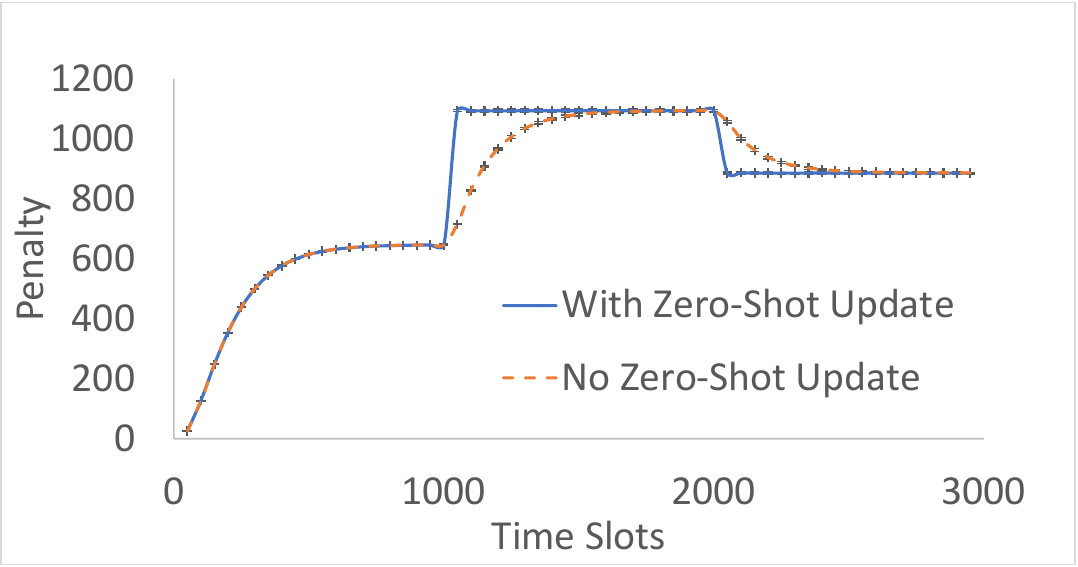}
       \label{fig:primal2_a_penalty}
    }
    \subfigure[Dynamics in user composition]
    {
       \includegraphics[width=0.3\linewidth]{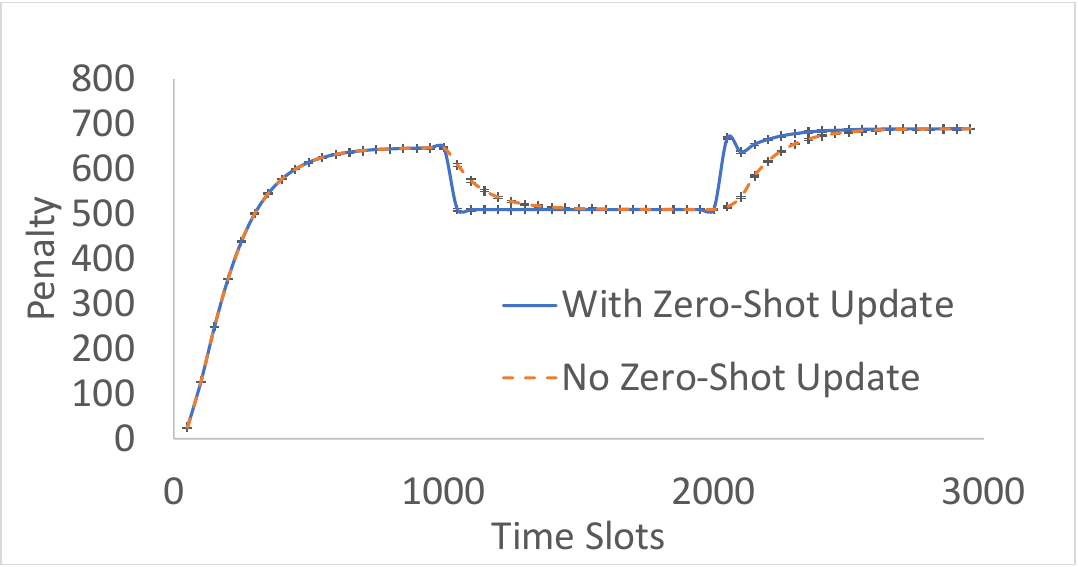}
       \label{fig:primal2_n_penalty}
    }
    \subfigure[Dynamics in channel state distribution]
    {
       \includegraphics[width=0.3\linewidth]{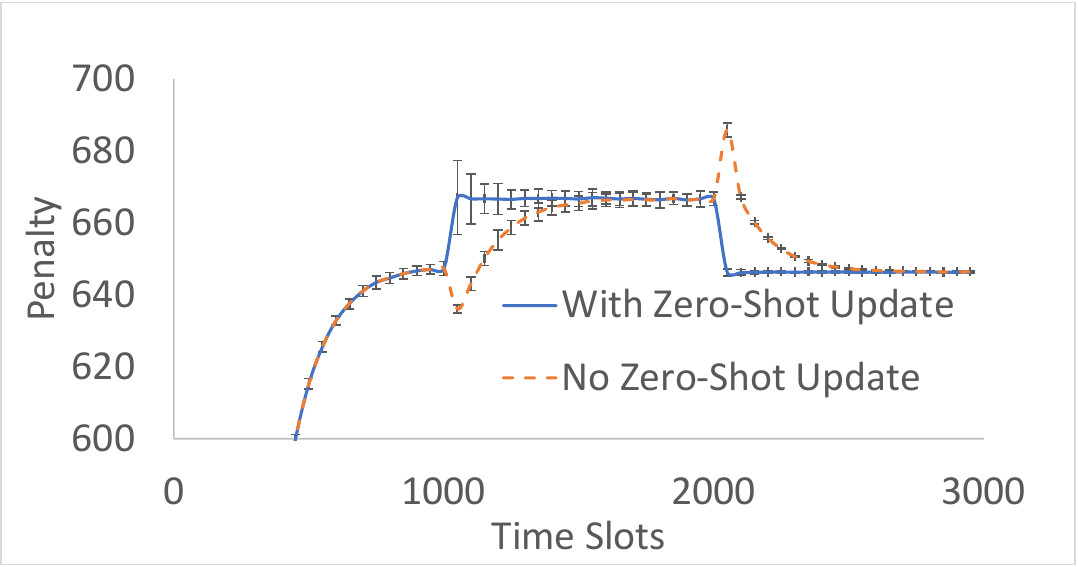}
       \label{fig:primal2_f_penalty}
    }
    \end{center}
    \caption{Total penalty for \textbf{Primal2} under various system dynamics.}
    \label{fig:primal2_penalty}
\end{figure*}

\begin{figure*}[t]
    \begin{center}
    \subfigure[Dynamics in service requirements]
    {
       \includegraphics[width=0.3\linewidth]{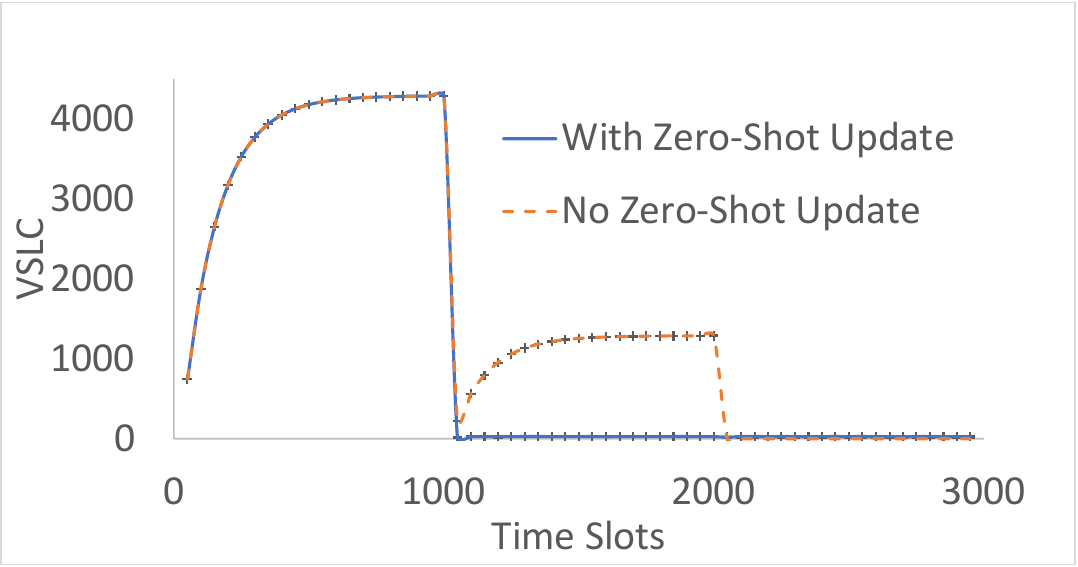}
       \label{fig:primal2_a_vslc}
    }
    \subfigure[Dynamics in user composition]
    {
       \includegraphics[width=0.3\linewidth]{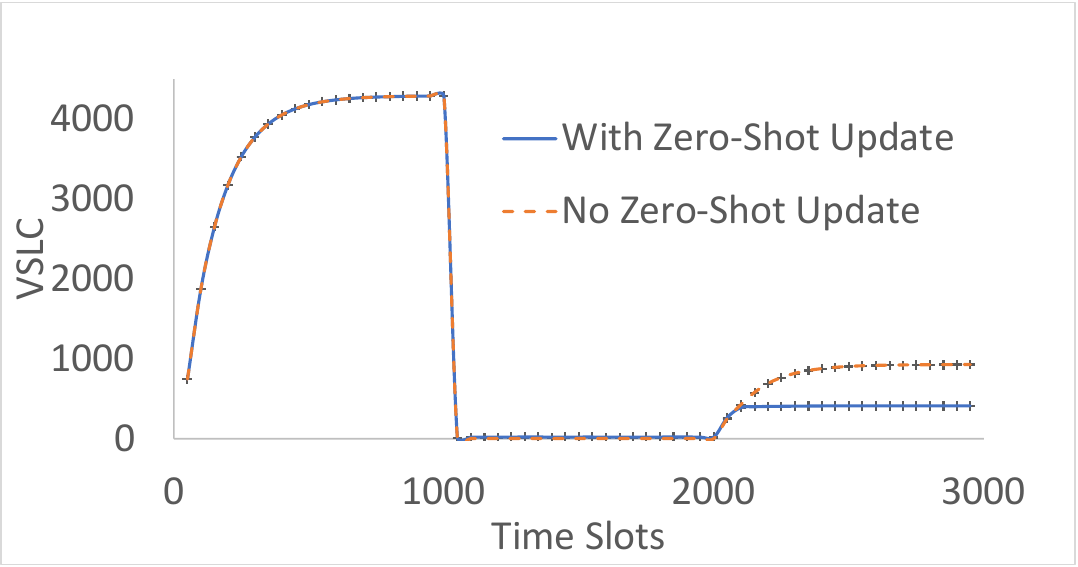}
       \label{fig:primal2_n_vslc}
    }
    \subfigure[Dynamics in channel state distribution]
    {
       \includegraphics[width=0.3\linewidth]{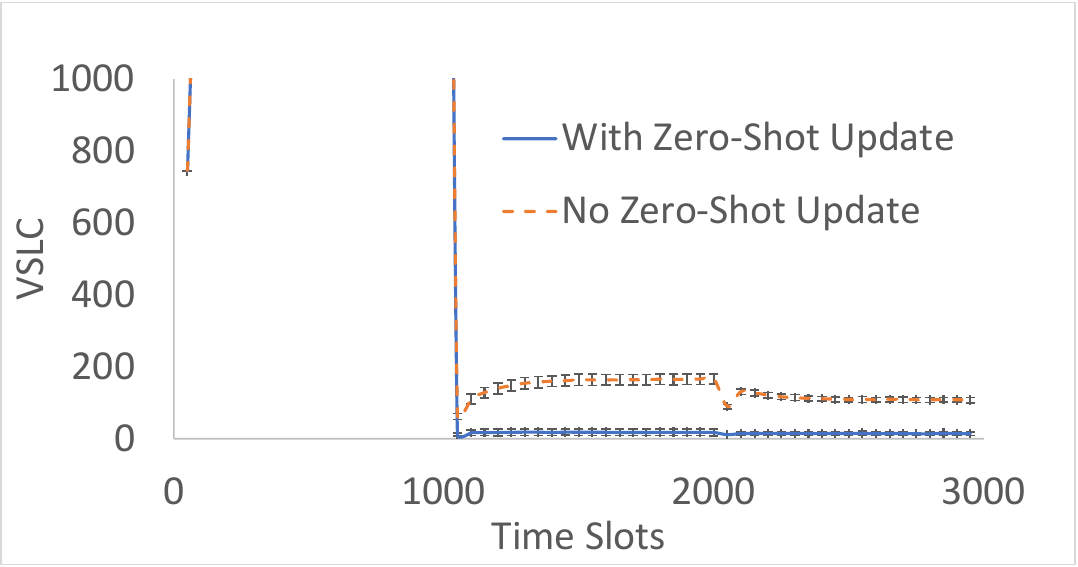}
       \label{fig:primal2_f_vslc}
    }
    \end{center}
    \caption{Total violation in service requirements for \textbf{Primal2} under various system dynamics.}
    \label{fig:primal2_vslc}
\end{figure*}

We consider a system with $N$ users and five channel states, numbered as $s=1,2, \dots,5$. The penalty function is $P_s(\vec{x}):=\frac{1}{2}(\sum_n x_n)^2+\sum_n\frac{1}{2\gamma_{n,s}}x_n^2$, where $\gamma_{n,s}=0.2+0.1\times((s+n+5)\mod 7)$.  When the system starts at time 0, $N=11$, each user $n$ has a service requirement $A_n=1.8+0.2n$, and each state $s$ occurs with probability $f_s=0.2$. 

We will evaluate the performance of our zero-shot update in three scenarios of system dynamics: dynamics in service requirements, dynamics in user composition, and dynamics in state distribution. For each scenario, we evaluate the total 
penalty incurred and the total violation in service requirements. To measure the violation in service requirements, we define violation-since-last-change (VSLC) as follows: Suppose that the last time any aspect of the system changes is at time $T$, then the VSLC at time $t$ is $\sum_n \max\{0,\sum_{\tau=T+1}^t x_{n,\tau}-(t-T)A_n\}$. Due to the randomness of the channel state, we record the running average of penalty and VSLC over the previous 50 time slots. We evaluate each scenario over 100 independent runs and report their mean and standard deviation.\\

\balance

\noindent\textbf{Dynamics in service requirements.} In this scenario, all service requirements $A_n$ increase by $30\%$ at time 1000 and decrease by $10\%$ at time 2000. Simulation results are shown in Fig. \ref{fig:primal2_a_penalty} and \ref{fig:primal2_a_vslc}. When the zero-shot update is employed, the total penalty is within $0.1\%$ of the optimal value immediately without any additional updates and the VSLC never exceeds 30 after each change in service requirements. On the other hand, without the zero-shot update, the total penalty remains more than $1\%$ different from the optimal value after more than 400 slots after each change in service requirements. Moreover, when $A_n$ increase by $30\%$, the system without the zero-shot update needs to accumulates more than 1200 VSLC before it can converge to the optimal value.\\

\noindent\textbf{Dynamics in user composition.} In this scenario, user 11 leaves the system at time 1000 and another user with $A_n=5$ joins the system at time 2000. Simulation results are shown in Figs. \ref{fig:primal2_n_penalty} and \ref{fig:primal2_a_vslc}. When the zero-shot update is employed, the total penalty is within $0.1\%$ of the optimal value immediately when user 11 leaves the system. When another user joins the system, it takes a small amount of time for the system with the zero-shot update to converge. This is because we are not using $\gamma_{s,n}$ of the new user in making the zero-shot update. Still, the system with the zero-shot update converges much faster and accumulates much less VSLC than the system without the zero-shot update.\\

\noindent\textbf{Dynamics in state distribution.} In this scenario, the vector $\vec{f}$ changes from $[0.2, 0.2, 0.2, 0.2, 0.2]$ to $[0.5, 0.5, 0, 0, 0]$ at time 1000 and then changes to $[0,0,0.2, 0.4, 0.4]$ at time 2000. Simulation results are shown in Figs. \ref{fig:primal2_f_penalty} and \ref{fig:primal2_f_vslc}. Similar to the previous two scenarios, it can be observed that the system with the zero-shot update converges to the optimal penalty almost immediately and accumulates much less VSLC than the system without the zero-shot update.

\section{Conclusion}\label{sec:conclusion}

This paper considers network optimization problems in dynamic systems. It studies two different network optimization problems: distributed rate control for utility maximization and stochastic network optimization with service requirements. For each problem, the paper considers that all aspects of the system, including user composition, system capacity, service requirements, and channel qualities, can change abruptly. In order to enable fast adaptation to any kind of change, the paper derives the first-order approximations of the optimal Lagrange multipliers with respect to any parameters of the system. Using these first-order approximations, the paper proposes zero-shot online learning algorithms for the two network optimization problems. The simulation results show that the zero-shot online learning algorithms significantly improve system performance in the transitory phases.

\newpage
\bibliographystyle{ieeetr}
\bibliography{ref}

\begin{thebibliography}{10}

\bibitem{tassiulas1990stability}
L.~Tassiulas and A.~Ephremides, ``Stability properties of constrained queueing
  systems and scheduling policies for maximum throughput in multihop radio
  networks,'' in {\em 29th IEEE Conference on Decision and Control},
  pp.~2130--2132, IEEE, 1990.

\bibitem{kelly1997charging}
F.~Kelly, ``Charging and rate control for elastic traffic,'' {\em European
  transactions on Telecommunications}, vol.~8, no.~1, pp.~33--37, 1997.

\bibitem{kelly1998rate}
F.~P. Kelly, A.~K. Maulloo, and D.~K.~H. Tan, ``Rate control for communication
  networks: shadow prices, proportional fairness and stability,'' {\em Journal
  of the Operational Research society}, vol.~49, no.~3, pp.~237--252, 1998.

\bibitem{stolyar2005maximizing}
A.~L. Stolyar, ``Maximizing queueing network utility subject to stability:
  Greedy primal-dual algorithm,'' {\em Queueing Systems}, vol.~50,
  pp.~401--457, 2005.

\bibitem{lin2006tutorial}
X.~Lin, N.~B. Shroff, and R.~Srikant, ``A tutorial on cross-layer optimization
  in wireless networks,'' {\em IEEE Journal on Selected areas in
  Communications}, vol.~24, no.~8, pp.~1452--1463, 2006.

\bibitem{palomar2006tutorial}
D.~P. Palomar and M.~Chiang, ``A tutorial on decomposition methods for network
  utility maximization,'' {\em IEEE Journal on Selected Areas in
  Communications}, vol.~24, no.~8, pp.~1439--1451, 2006.

\bibitem{neely2022stochastic}
M.~Neely, {\em Stochastic network optimization with application to
  communication and queueing systems}.
\newblock Springer Nature, 2022.

\bibitem{yi2008stochastic}
Y.~Yi and M.~Chiang, ``Stochastic network utility maximisation—a tribute to
  kelly's paper published in this journal a decade ago,'' {\em European
  Transactions on Telecommunications}, vol.~19, no.~4, pp.~421--442, 2008.

\bibitem{wang2018distributed}
Y.~Wang, W.~Wang, Y.~Cui, K.~G. Shin, and Z.~Zhang, ``Distributed packet
  forwarding and caching based on stochastic network utility maximization,''
  {\em IEEE/ACM Transactions on Networking}, vol.~26, no.~3, pp.~1264--1277,
  2018.

\bibitem{hou2010utilitya}
I.-H. Hou and P.~Kumar, ``Utility maximization for delay constrained qos in
  wireless,'' in {\em 2010 Proceedings IEEE INFOCOM}, pp.~1--9, IEEE, 2010.

\bibitem{hou2010utilityb}
I.-H. Hou and P.~Kumar, ``Utility-optimal scheduling in time-varying wireless
  networks with delay constraints,'' in {\em Proceedings of the eleventh ACM
  international symposium on Mobile ad hoc networking and computing},
  pp.~31--40, 2010.

\bibitem{zuo2017energy}
S.~Zuo, H.~Deng, and I.-H. Hou, ``Energy efficient algorithms for real-time
  traffic over fading wireless channels,'' {\em IEEE Transactions on Wireless
  Communications}, vol.~16, no.~3, pp.~1881--1892, 2017.

\bibitem{jiang2022joint}
H.~Jiang, X.~Dai, Z.~Xiao, and A.~Iyengar, ``Joint task offloading and resource
  allocation for energy-constrained mobile edge computing,'' {\em IEEE
  Transactions on Mobile Computing}, vol.~22, no.~7, pp.~4000--4015, 2022.

\bibitem{xia2021online}
S.~Xia, Z.~Yao, Y.~Li, and S.~Mao, ``Online distributed offloading and
  computing resource management with energy harvesting for heterogeneous
  mec-enabled iot,'' {\em IEEE Transactions on Wireless Communications},
  vol.~20, no.~10, pp.~6743--6757, 2021.

\bibitem{zhang2013robust}
Y.~Zhang, N.~Gatsis, and G.~B. Giannakis, ``Robust energy management for
  microgrids with high-penetration renewables,'' {\em IEEE transactions on
  sustainable energy}, vol.~4, no.~4, pp.~944--953, 2013.

\bibitem{wang2018joint}
J.~Wang, C.~Jiang, Z.~Wei, C.~Pan, H.~Zhang, and Y.~Ren, ``Joint uav hovering
  altitude and power control for space-air-ground iot networks,'' {\em IEEE
  Internet of Things Journal}, vol.~6, no.~2, pp.~1741--1753, 2018.

\bibitem{ma2021uav}
T.~Ma, H.~Zhou, B.~Qian, N.~Cheng, X.~Shen, X.~Chen, and B.~Bai, ``Uav-leo
  integrated backbone: A ubiquitous data collection approach for b5g internet
  of remote things networks,'' {\em IEEE Journal on Selected Areas in
  Communications}, vol.~39, no.~11, pp.~3491--3505, 2021.

\bibitem{liu2016achieving}
J.~Liu, ``Achieving low-delay and fast-convergence in stochastic network
  optimization: A nesterovian approach,'' {\em ACM SIGMETRICS Performance
  Evaluation Review}, vol.~44, no.~1, pp.~221--234, 2016.

\bibitem{liu2016heavy}
J.~Liu, A.~Eryilmaz, N.~B. Shroff, and E.~S. Bentley, ``Heavy-ball: A new
  approach to tame delay and convergence in wireless network optimization,'' in
  {\em IEEE INFOCOM 2016-The 35th Annual IEEE International Conference on
  Computer Communications}, pp.~1--9, IEEE, 2016.

\bibitem{chen2017learn}
T.~Chen, Q.~Ling, and G.~B. Giannakis, ``Learn-and-adapt stochastic dual
  gradients for network resource allocation,'' {\em IEEE Transactions on
  Control of Network Systems}, vol.~5, no.~4, pp.~1941--1951, 2017.

\bibitem{fu2021learning}
X.~Fu and E.~Modiano, ``Learning-num: Network utility maximization with unknown
  utility functions and queueing delay,'' in {\em Proceedings of the
  Twenty-second International Symposium on Theory, Algorithmic Foundations, and
  Protocol Design for Mobile Networks and Mobile Computing}, pp.~21--30, 2021.

\bibitem{verma2020stochastic}
A.~Verma and M.~K. Hanawal, ``Stochastic network utility maximization with
  unknown utilities: Multi-armed bandits approach,'' in {\em IEEE INFOCOM
  2020-IEEE Conference on Computer Communications}, pp.~189--198, IEEE, 2020.

\bibitem{liang2018network}
Q.~Liang and E.~Modiano, ``Network utility maximization in adversarial
  environments,'' in {\em IEEE INFOCOM 2018-IEEE Conference on Computer
  Communications}, pp.~594--602, IEEE, 2018.

\bibitem{yang2023learning}
Z.~Yang, R.~Srikant, and L.~Ying, ``Learning while scheduling in multi-server
  systems with unknown statistics: Maxweight with discounted ucb,'' in {\em
  International Conference on Artificial Intelligence and Statistics},
  pp.~4275--4312, PMLR, 2023.

\bibitem{nguyen2023learning}
Q.~M. Nguyen and E.~Modiano, ``Learning to schedule in non-stationary wireless
  networks with unknown statistics,'' in {\em Proceedings of the Twenty-fourth
  International Symposium on Theory, Algorithmic Foundations, and Protocol
  Design for Mobile Networks and Mobile Computing}, pp.~181--190, 2023.

\bibitem{lin2004joint}
X.~Lin and N.~B. Shroff, ``Joint rate control and scheduling in multihop
  wireless networks,'' in {\em 2004 43rd IEEE Conference on Decision and
  Control (CDC)(IEEE Cat. No. 04CH37601)}, vol.~2, pp.~1484--1489, IEEE, 2004.

\bibitem{mhanna2018component}
S.~Mhanna, A.~C. Chapman, and G.~Verbi{\v{c}}, ``Component-based dual
  decomposition methods for the opf problem,'' {\em Sustainable Energy, Grids
  and Networks}, vol.~16, pp.~91--110, 2018.

\bibitem{notarnicola2017distributed}
I.~Notarnicola, R.~Carli, and G.~Notarstefano, ``Distributed partitioned
  big-data optimization via asynchronous dual decomposition,'' {\em IEEE
  Transactions on Control of Network Systems}, vol.~5, no.~4, pp.~1910--1919,
  2017.

\bibitem{huang2014power}
L.~Huang, X.~Liu, and X.~Hao, ``The power of online learning in stochastic
  network optimization,'' in {\em The 2014 ACM international conference on
  Measurement and modeling of computer systems}, pp.~153--165, 2014.

\end{thebibliography}

\end{document}